\documentclass[11pt,letterpaper,draftcls,oneside,peerreview,onecolumn]{IEEEtran}
\usepackage{amsfonts}
\usepackage{amssymb}
\usepackage{amsmath}
\usepackage[dvips]{graphicx}
\usepackage{cite}
\usepackage{hyperref}
\usepackage[left=0.8in,top=1.0in,bottom=1.2in,right=0.8in]{geometry}

\newtheorem{theorem}{Theorem}

\newtheorem{corollary}{Corollary}

\newtheorem{lemma}[theorem]{Lemma}

\newtheorem{proposition}{Proposition}

\begin{document}

\title{Protocols for Relay-Assisted Free-Space Optical Systems}
\author{Nestor~D.~Chatzidiamantis,~\IEEEmembership{Student~Member,~IEEE,}
~Diomidis~S.~Michalopoulos, \IEEEmembership{Member,~IEEE,}~Emmanouil E.
Kriezis,~\IEEEmembership{Member,~IEEE,}~George K. Karagiannidis, %
\IEEEmembership{Senior~Member,~IEEE,}~and~Robert Schober,~%
\IEEEmembership{Fellow,~IEEE}\thanks{%
This paper has been submitted in part to the IEEE Global Communications
Conference (GLOBECOM'11).}\thanks{%
N. D. Chatzidiamantis, E. E. Kriezis, and G. K. Karagiannidis are with the
Department of Electrical and Computer Engineering, Aristotle University of
Thessaloniki, GR-54124 Thessaloniki, Greece (E-mail: \{nestoras, mkriezis,
geokarag\}@auth.gr).} \thanks{%
D. S. Michalopoulos and R. Schober are with the Department of Electrical and
Computer Engineering, The University of British Columbia, Vancouver, BC V6T
1Z4, Canada (E-mail: \{dio, rschober\}@ece.ubc.ca).}}
\pubid{}
\specialpapernotice{}
\maketitle

\begin{abstract}
We investigate transmission protocols for relay-assisted free-space optical
(FSO) systems, when multiple parallel relays are employed and there is no
direct link between the source and the destination. As alternatives to
all-active FSO relaying, where all the available relays transmit
concurrently, we propose schemes that select only a single relay to
participate in the communication between the source and the destination in
each transmission slot. This selection is based on the channel state
information (CSI) obtained either from all or from some of the FSO links.
Thus, the need for synchronizing the relays' transmissions is avoided and
the slowly varying nature of the atmospheric channel is exploited. For both
relay selection and all-active relaying, novel closed-form expressions for
their outage performance are derived, assuming the versatile Gamma-Gamma
channel model. Furthermore, based on the derived analytical results, the
problem of allocating the optical power resources to the FSO links is
addressed, and optimum and suboptimum solutions are proposed. Numerical
results are provided for equal and non-equal length FSO links, which
illustrate the outage behavior of the considered relaying protocols and
demonstrate the significant performance gains offered by the proposed power
allocation schemes.
\end{abstract}

\begin{keywords}
Atmospheric turbulence, cooperative diversity, distributed switch and stay
relaying, free-space optical communications, relay-assisted communications,
relay selection, power allocation.
\end{keywords}

\markboth{Submitted to the \textit{IEEE Transactions on
Communications}}
{Murray and Balemi: Using the Document Class IEEEtran.cls}%
\setcounter{page}{0}\newpage

\section{Introduction}

The constant need for higher data rates in support of high-speed
applications has led to the development of the Free Space Optical (FSO)
communication technology. Operating at unlicensed optical frequencies, FSO
systems offer the potential of broadband capacity at low cost \cite%
{B:Willebrand}, and therefore, they present an attractive remedy for the
"last-mile" problem. However, despite their major advantages, the widespread
deployment of FSO systems is hampered by major impairments, which have their
origin in the propagation of optical signals through the atmosphere. Rain,
fog, and atmospheric turbulence are some of the major atmospheric phenomena
that cause attenuation and rapid fluctuations in the received optical power
in FSO systems, thereby increasing the error rate and severely degrading the
overall performance \cite{B:Andrews}.

In the past, several techniques have been applied in FSO systems for
mitigating the degrading effects of the atmospheric channel, including error
control coding in conjunction with interleaving \cite{J:Kah4},
multiple-symbol detection \cite{J:SchoberTCOM}, and spatial diversity \cite%
{J:Lee,J:Navid,J:Wilson_FSO_MIMO}. Among these techniques, spatial
diversity, which is realized by deploying multiple transmit and/or receive
apertures, has been particularly attractive, since it offers significant
performance gains by introducing additional degrees of freedom in the
spatial dimension. Thus, numerous FSO systems with multiple co-located
transmit and/or receive apertures, referred as Multiple-Input
Multiple-Output (MIMO) FSO systems, have been proposed in the technical
literature \cite{J:Lee,J:Navid,J:Wilson_FSO_MIMO}. However, in practice,
MIMO FSO systems may not always be able to offer the gains promised by
theory. This happens in cases where the assumption that all the links of the
MIMO FSO system are affected by independent channel fading becomes invalid
\cite{J:Wilson_FSO_MIMO}. Furthermore, since both the path loss and the
fading statistics of the channel are distance-dependent, a large number of
transmit and/or receive apertures is required in long-range links in order
to achieve the desired performance gains, thus increasing the complexity of
MIMO FSO systems.

In order to overcome such limitations, relay-assisted communication has been
recently introduced in FSO systems as an alternative approach to achieve
spatial diversity \cite{J:UysalRelays,J:Kamiri,J:Kamiri2,J:Rjeily}. The main
idea lies in the fact that, by employing multiple relay nodes with
line-of-sight (LOS) to both the source and the destination, a virtual
multiple-aperture system is created, often referred as cooperative diversity
system, even if there is no LOS between the source and the destination. In
\cite{J:UysalRelays}, various relaying configurations (cooperative diversity
and multihop) have been investigated under the assumption of a lognormal
channel model. Subsequently, several coding schemes for 3-way cooperative
diversity FSO systems with a single relay and a direct link between the
source and the destination have been proposed in \cite{J:Kamiri}, while the
performance of such systems has been investigated in \cite{J:Kamiri2} and
\cite{J:Rjeily} assuming amplify-and-forward and decode-and-forward relaying
strategies, respectively. It is emphasized that in all these previous works,
all the available relays participated in the communication between the
source and the destination, requiring perfect synchronization between the
relays such that the FSO signals can arrive simultaneously at the
destination, while the optical power resources are equally divided between
all FSO links.

In view of the above, in this paper we present alternative transmission
protocols which can be applied to relay-assisted FSO systems with no LOS
between the source and the destination. For the signaling rates of interest,
the atmospheric channel does not vary within one packet. Thus, channel state
information (CSI) can be easily obtained for all or for some of the involved
links. Capitalizing on this fact, the presented protocols select only a
single relay to take part in the communication in every transmission slot,
thus avoiding the need for synchronization between the relays. It should be
noted that similar relay selection protocols have been also proposed in the
context of radio-frequency relaying systems \cite{J:KarRel,J:Dio,J:Dio2}. In
particular, two types of relay selection protocols are presented: the
select-max protocol that selects the relay that maximizes an appropriately
defined metric and requires CSI from all the available FSO links, and the
distributed switch and stay (DSSC) protocol which switches between two
relays and requires CSI only from the FSO links used in the previous
transmission slot. Furthermore, assuming the versatile Gamma-Gamma channel
model \cite{B:Andrews} and decode-and-forward relay nodes, we derive novel
closed-form analytical expressions for the outage performance of the
proposed transmission schemes, as well as the scheme where all the available
relays transmit simultaneously; thus, extending the analysis presented in
\cite{J:UysalRelays} to the case of the Gamma-Gamma channel model. Finally,
based on the derived outage results, we address the problem of optimizing
the allocation of the optical power resources to the FSO links for
minimization of the probability of outage; hence, rendering the
relay-assisted FSO system under consideration more power efficient.

The remainder of the paper is organized as follows. In Section \ref{SM}, the
system model and the considered relaying protocols are discussed. The outage
performance of the relaying protocols under investigation is analyzed in
Section \ref{OA}, while the problem of optimizing the allocation of the
optical power resources to the FSO links is addressed in Section \ref{OPA}.
Numerical results for various relay-assisted FSO architectures are presented
in Section \ref{NR}, and, finally, concluding remarks are provided in
Section \ref{Con}.

\section{System Model\label{SM}}

The system model under consideration is depicted in Fig. \ref{Fig:system}.
In particular, we consider an intensity-modulation direct detection (IM/DD)
FSO system without LOS between the source, $S$, and the destination, $D$,
and the communication between these two terminals is achieved with the aid
of multiple relays, denoted by $R_{i}$, $i\in \left\{ 1,...,N\right\} $. The
source node is equipped with a multiple-aperture transmitter, with each of
the apertures pointing in the direction of the corresponding relay, and an
optical switch\footnote{%
Optical switches can be implemented with either spatial light modulators
(SLM) \cite[Ch. (27)]{B:Kriezis} or optical MEMS devices \cite{J:mems}.},
which either allows the simultaneous transmission from all the transmit
apertures or selects the direction of transmission by switching between the
transmit apertures.

The presence of a large field-of-view (FOV) detector at the destination is
assumed allowing for the simultaneous detection of the optical signals
transmitted from each relay. Moreover, all optical transmitters are equipped
with optical amplifiers that adjust the optical power transmitted in each
link. The relaying terminals use a threshold-based decode-and-forward (DF)
protocol; that is, they fully decode the received signal and retransmit it
to the destination only if the signal-to-noise ratio (SNR) of the receiving
FSO link exceeds a given decoding threshold. Finally, throughout this paper,
we assume that binary pulse position modulation (BPPM) is employed.

\subsection{Signal and Channel Model}

For an FSO link connecting two terminals $A$ and $B$, the received optical
signal at the photodetector of $B$ is given by%
\begin{equation}
\mathbf{r}_{B}=\left[
\begin{array}{c}
r^{s} \\
r^{n}%
\end{array}%
\right] =\left[
\begin{array}{c}
\eta T_{b}\left( \rho _{AB}P_{t}h_{_{AB}}+P_{b}\right) +n^{s} \\
\eta T_{b}P_{b}+n^{n}%
\end{array}%
\right]  \label{rec_signal}
\end{equation}%
where $r^{s}$ and $r^{n}$ represent the signal and the non-signal slots of
the BPPM symbol, respectively, while $\rho _{AB}P_{t}$ and $P_{b}$ denote
the average optical signal power transmitted from $A$ and the background
radiation incident on the photodetector of $B$, respectively. Furthermore, $%
\rho _{AB}$ represents the percentage of the total optical power $P_{t}$
allocated to the FSO link between terminals $A$ and $B$, $h_{_{AB}}$ is the
channel gain of the link, $\eta $ is the photodetector's responsivity, $%
T_{b} $ is the duration of the signal and non-signal slots, and $n^{s}$ and $%
n^{n}$ are the additive noise samples in the signal and non-signal slots,
respectively. Since background-noise limited receivers are assumed, where
background noise is dominant compared to other noise components (such as
thermal, signal dependent, and dark noise) \cite{J:Lee,J:Kah1}, the noise
terms can be modeled as additive white Gaussian, with zero mean and variance
$\sigma _{n}^{2}=\frac{N_{0}}{2}$. After removing the constant bias $\eta
T_{b}P_{b}$ from both slots, the instantaneous SNR of the link can be
defined as \cite{J:UysalRelays}%
\begin{equation}
\gamma _{_{AB}}=\frac{\eta ^{2}\rho _{AB}^{2}T_{b}^{2}P_{t}^{2}h_{_{AB}}^{2}%
}{N_{0}}.  \label{inst_SNR}
\end{equation}

Due to atmospheric effects, the channel gain of the FSO link under
consideration can be modeled as%
\begin{equation}
h_{_{AB}}=\bar{h}_{_{AB}}\tilde{h}_{_{AB}}  \label{chan_state}
\end{equation}%
where $\bar{h}_{AB}$ accounts for path loss due to weather effects and
geometric spread loss and $\tilde{h}_{_{AB}}$ represents irradiance
fluctuations caused by atmospheric turbulence. Both $\bar{h}_{_{AB}}$ and $%
\tilde{h}_{_{AB}}$ are time-variant, yet at very different time scales. The
path loss coefficient varies on the order of hours while turbulence induced
fading varies on the order of 1--100 ms \cite{J:Lee}. Thus, taking into
consideration the signaling rates of interest, which range from hundreds to
thousands of Mbps, the channel gain can be considered as constant over a
given transmission slot, which consists of hundreds of thousands (or even
millions) of consecutive symbols.

The path loss coefficient can be calculated by combining the Beer Lambert's
law \cite{B:Andrews} with the geometric loss formula \cite[pp. 44]%
{B:Willebrand}, yielding%
\begin{equation}
\bar{h}_{_{AB}}=\frac{D_{R}^{2}}{\left( D_{T}+\theta _{T}d_{_{AB}}\right)
^{2}}\exp \left( -vd_{_{AB}}\right)  \label{geom_loss}
\end{equation}%
where $D_{R}$ and $D_{T}$ are the receiver and transmitter aperture
diameters, respectively; $\theta _{T}$ is the optical beam's divergence
angle (in $m$rad), $d_{_{AB}}$ is the link's distance (in km), and $v$ is
the weather dependent attenuation coefficient (in 1/km).

Under a wide range of atmospheric conditions, turbulence induced fading can
be statistically characterized by the well-known Gamma-Gamma distribution
\cite{B:Andrews}. The probability density function (pdf) for this model is
given by
\begin{equation}
f_{\tilde{h}_{_{AB}}}\left( x\right) =\frac{2\left( \alpha _{_{AB}}\beta
_{_{AB}}\right) ^{\frac{\alpha _{_{AB}}+\beta _{_{AB}}}{2}}}{\Gamma (\alpha
_{_{AB}})\Gamma (\beta _{_{AB}})}x^{\frac{\alpha _{_{AB}}+\beta _{_{AB}}}{2}%
-1}K_{\alpha _{_{AB}}-\beta _{_{AB}}}\left( 2\sqrt{\alpha _{_{AB}}\beta
_{_{AB}}x}\right)  \label{pdf}
\end{equation}%
where $\Gamma \left( \cdot \right) $ is the Gamma function \cite[Eq. (8.310)]%
{B:Gra_Ryz_Book} and $K_{\nu }\left( \cdot \right) $ is the $\nu $th order
modified Bessel function of the second kind \cite[Eq. (8.432/9)]%
{B:Gra_Ryz_Book}. Furthermore, $\alpha _{AB}$ and $\beta _{AB}$ are
parameters related to the effective atmospheric conditions via $\alpha
_{_{AB}}=\left[ \exp \left( \frac{0.49\sigma _{R}^{2}}{\left( 1+1.11\sigma
_{R}^{\frac{12}{5}}\right) ^{\frac{7}{6}}}\right) -1\right] ^{-1}$and $\beta
_{_{AB}}=\left[ \exp \left( \frac{0.51\sigma _{R}^{2}}{\left( 1+0.69\sigma
_{R}^{\frac{12}{5}}\right) ^{\frac{7}{6}}}\right) -1\right] ^{-1}$ \cite%
{B:Andrews}, where $\sigma _{R}^{2}=1.23C_{n}^{2}\left( \frac{2\pi }{\lambda
}\right) ^{\frac{7}{6}}d_{_{AB}}^{\frac{11}{6}}$ denotes the Rytov variance%
\footnote{%
The Rytov variance is indicative of the strength of turbulence-induced
fading. More specifically, values $\sigma _{R}^{2}<1$ correspond to weak
turbulence conditions, while values $\sigma _{R}^{2}>1$ correspond to the
moderate-strong turbulence regime \cite{J:Andrews_Gamma}.}, $C_{n}^{2}$ is
the weather dependent index of refraction structure parameter, and $\lambda $
represents the wavelength of the optical carrier.

\subsection{Mode of Operation}

Throughout this work, three different cooperative relaying protocols are
considered: the all-active protocol, originally presented in \cite%
{J:UysalRelays}, where all the available relays are activated, and the
select-max and the distributed switch and stay combining (DSSC) protocols,
which are both based on the concept of selecting a single relay.

\subsubsection{All-active}

In this relaying scheme, the source activates all relays and the total power
is divided between all available FSO links. Since the relay nodes operate in
the DF mode only the relays that successfully decode the received optical
signal remodulate the intensity of the optical carrier and forward the
information to the destination. At the destination, owing to the presence of
a large FOV aperture, aperture averaging occurs \cite{J:Navid} and all the
received optical signals are added. Hence, assuming perfect synchronization,
the output of the combiner can be expressed as%
\begin{equation}
\mathbf{r}_{D}=\left[
\begin{array}{c}
\eta T_{b}\left( \sum\limits_{m\in \mathbf{D}}\rho
_{_{R_{m}D}}h_{_{R_{m}D}}P_{t}+P_{b}\right) +n^{s} \\
\eta T_{b}P_{b}+n^{n}%
\end{array}%
\right]  \label{rec_singal_EGC}
\end{equation}%
where $\mathbf{D}$ denotes the decoding set formed by the relays that have
succesfully decoded the signal. Since the total power is divided between all
available links, it follows that $\sum_{i=1}^{N}\left( \rho
_{_{SR_{i}}}+\rho _{_{R_{i}D}}\right) =1.$

The advantage of this scheme is that CSI is not required neither at the
transmitter nor the receiver side, since the source transmits to all
available\ relays, regardless of their channel gain. However, since it is
assumed that all the signals arrive at the destination at the same time,
this scheme requires accurate timing synchronization in order to account for
the different propagation delays of the different paths, resulting in high
complexity.

\subsubsection{Select-Max}

This relaying protocol selects a single relay out of the set of $N$
available relays in each transmission slot. In particular, the relay which
maximizes an appropriately defined metric is selected. This metric accounts
for both the $S$-$R_{i}$ and $R_{i}$-$D$ links and reflects the quality of
the $i$th end-to-end path. Here, we adopt the minimum value of the
intermediate link SNRs,%
\begin{equation}
\gamma _{i}=\min \left( \gamma _{_{SR_{i}}},\gamma _{_{R_{i}D}}\right) ,
\label{min}
\end{equation}%
as the quality measure of the $i$th end-to-end path, which will be referred
as the "\textit{min equivalent SNR}" throughout the paper. Note that (\ref%
{min}) represents an outage-based definition of the selection metric, in the
sense that an outage on the $i$th end-to-end link occurs if $\gamma _{i}$
falls below the outage threshold SNR. Hence, the single relay that is
activated in the select-max relaying protocol, $R_{b},$ is selected
according to the rule%
\begin{equation}
b=\operatornamewithlimits{argmax}_{i\in \left\{ 1,...N\right\} }\gamma _{i}.
\label{sel_max}
\end{equation}

Since a single relay is activated in the select-max protocol, the total
available optical power is divided between the $S$-$R_{b}$ and $R_{b}$-$D$
links, i.e., $\rho _{_{SR_{b}}}+\rho _{_{R_{b}D}}=1,$and in the case that $%
R_{_{b}}$ has successfully decoded the received optical signal, i.e., $b\in
\mathbf{D}$, the signal at the destination can be expressed as
\begin{equation}
\mathbf{r}_{D}=\left[
\begin{array}{c}
\eta T_{b}\left( \rho _{_{R_{b}D}}P_{t}h_{_{R_{b}D}}+P_{b}\right) +n^{s} \\
\eta T_{b}P_{b}+n^{n}%
\end{array}%
\right] .  \label{rec_sel_max}
\end{equation}

This relaying scheme requires the CSI of all the available $S$-$R_{i}$ and $%
R_{i}$-$D$ FSO links in order to perform the selection process. This can be
achieved by some signalling process that takes advantage of the
slowly-varying nature of the FSO channel . Here, each receiver estimates the
correspponding link CSI and feeds it back to the source through a reliable
low-rate RF feedback link. It is emphasized that since only one end-to-end
path is activated in each transmission slot, only one signal arrives at the
destination and thus synchronization between the relays is not needed.

\subsubsection{DSSC}

Requiring less CSI than select-max, the DSSC protocol applies to the case
where there are only two relays available and one of them is selected to
take part in the communication between the source and the destination, in a
switch-and-stay fashion \cite{J:Dio}. More specifically, in each
transmission slot the destination compares the $\min $ equivalent SNR of the
active end-to-end path with a switching threshold, denoted by $T$. If this
SNR is smaller than $T$, the destination notifies the source and the other
available relay is selected for taking part in the communication, regardless
of \ its end-to-end performance metric.

Mathematically speaking, denoting the two available relays by $R_{1}$ and $%
R_{2}$ and the $\min $ equivalent SNR of the $i$th end-to-end path during
the $j$th transmission period by $\gamma _{i}^{j}$, the active relay in the $%
j$th transmission period, $R_{b}^{j}$, is determined as follows:%
\begin{equation}
\text{if }R_{b}^{j-1}=R_{1}\text{ then }R_{b}^{j}=\left\{
\begin{array}{c}
R_{1}\text{ when }\gamma _{1}^{j}\geq T \\
R_{2}\text{ when }\gamma _{1}^{j}<T%
\end{array}%
\right.  \label{DSSC_1}
\end{equation}%
and
\begin{equation}
\text{if }R_{b}^{j-1}=R_{2}\text{ then }R_{b}^{j}=\left\{
\begin{array}{c}
R_{2}\text{ when }\gamma _{2}^{j}\geq T \\
R_{1}\text{ when }\gamma _{2}^{j}<T%
\end{array}%
\right. .  \label{DSSC_2}
\end{equation}%
Hence, in the case that $R_{b}^{j}$ has successfully decoded the received
signal, the optical signal at the destination is given by
\begin{equation}
\mathbf{r}_{D}=\left[
\begin{array}{c}
\eta T_{b}\left( \rho _{_{R_{b}^{j}D}}P_{t}h_{_{R_{b}^{j}D}}+P_{b}\right)
+n^{s} \\
\eta T_{b}P_{b}+n^{n}%
\end{array}%
\right] .  \label{DSSC_rec}
\end{equation}%
Since in this protocol only a single relay assists in the communication
between the source and the destination, the power allocation rule of the
select-max protocol also holds for DSSC relaying.

When there are more than two available relays in the system, i.e., $N>2$, a
modified version of DSSC protocol could initially sort all the available
paths based on their end-to-end distance, defined as
\begin{equation}
d_{i}=\max \left( d_{_{SR_{i}}},d_{_{R_{i}D}}\right)  \label{DSSC_criterion}
\end{equation}%
with $i=1,...,N$, and, then, use as $R_{1}$ and $R_{2}$ the two relays that
correspond to the paths with the minimum end-to-end distance. It should be
noted that end-to-end distance is an indicative of the path's end-to-end
performance, taking into consideration that both path loss and Rytov
variance are monotonically increasing with respect to the link distance.

The simplicity of this scheme compared to the select-max protocol lies in
the fact that only the CSI of the active end-to-end path is required for the
selection process, resulting in less implementation complexity. Furthermore,
as in the select-max scheme, no synchronization among the relays is needed,
since only one end-to-end path is activated in each transmission slot.

\section{Outage Analysis\label{OA}}

At a given transmission rate, $r_{0},$ the outage probability is defined as
\begin{equation}
P_{out}\left( r_{0}\right) =\Pr \left\{ C\left( \gamma \right)
<r_{0}\right\} ,  \label{outage_1}
\end{equation}%
where $C\left( \cdot \right) $ is the instantaneous capacity, which is a
function of the instantaneous SNR. Since $C\left( \cdot \right) $ is
monotonically increasing with respect to $\gamma $, (\ref{outage_1}) can be
equivalently rewritten as
\begin{equation}
P_{out}\left( r_{0}\right) =\Pr \left\{ \gamma <\gamma _{th}\right\} ,
\label{outage_2}
\end{equation}%
where $\gamma _{th}=C^{-1}\left( r_{0}\right) $ denotes the threshold SNR.
If the SNR, $\gamma $, drops below $\gamma _{th}$, an outage occurs,
implying that the signal cannot be decoded with arbitrarily low error
probability at the receiver. Henceforth, it is assumed that the threshold
SNR, $\gamma _{th}$, is identical for all links of the relaying system.

\subsection{Outage Probability of the Intermediate Links}

Since DF relaying is considered, an outage event in any of the intermediate
links may lead to an outage of the overall relaying scheme. Therefore, the
calculation of the outage probability of each intermediate link is
considered as a building block for the outage probability of the relaying
schemes under investigation.

By combining (\ref{inst_SNR}) with (\ref{outage_2}), the outage probability
of the FSO link between nodes $A$ and $B$ is defined as%
\begin{equation}
P_{out,AB}=\Pr \left\{ \frac{\eta ^{2}T_{b}^{2}\rho
_{AB}^{2}P_{t}^{2}h_{_{AB}}^{2}}{N_{0}}<\gamma _{th}\right\}
\label{Pout_SISO}
\end{equation}%
which can be equivalently rewritten as%
\begin{equation}
P_{out,AB}=\Pr \left\{ \tilde{h}_{_{AB}}<\frac{1}{\bar{h}_{_{AB}}\rho
_{AB}P_{M}}\right\}  \label{Pout_SISO_3}
\end{equation}%
where $P_{M}$ is the power margin given by $P_{M}=\frac{\eta T_{b}P_{t}}{%
\sqrt{N_{0}\gamma _{th}}}$. Using the cumulative density function (cdf) of
the Gamma-Gamma distribution \cite[Eq. (7)]{J:Tsif_coh}, the outage
probability of the FSO link between nodes $A$ and $B$ can be analytically
evaluated for any $\alpha _{AB}$ and $\beta _{AB}$, yielding%
\begin{equation}
P_{out,AB}=\frac{1}{\Gamma \left( \alpha _{_{AB}}\right) \Gamma \left( \beta
_{_{AB}}\right) }G_{1,3}^{2,1}\left[ \frac{\frac{\alpha _{_{AB}}\beta
_{_{AB}}}{\bar{h}_{_{AB}}}}{P_{M}\rho _{AB}}\left\vert
\begin{array}{c}
1 \\
\alpha _{_{AB}},\beta _{_{AB}},0 \\
\end{array}%
\right. \right]  \label{Pout_SISO_Meijer}
\end{equation}%
where $G_{p,q}^{m,n}\left[ \cdot \right] $ is the Meijer's $G$-function \cite%
[Eq. (9.301)]{B:Gra_Ryz_Book}.

To gain more physical insights from (\ref{Pout_SISO_Meijer}), it is
meaningful to explore the outage probability in the high power margin regime.

\begin{theorem}
For high values of power margin and when $\left( \alpha _{_{AB}}-\beta
_{_{AB}}\right) \notin \mathbb{Z}$, the outage probability of the FSO link
between nodes $A$ and $B$ can be approximated by
\begin{equation}
P_{out,AB}\approx \frac{\Gamma \left( p_{_{AB}}-q_{_{AB}}\right) }{\Gamma
\left( \alpha _{_{AB}}\right) \Gamma \left( \beta _{_{AB}}\right) }\frac{%
\left( \frac{\alpha _{_{AB}}\beta _{_{AB}}}{\bar{h}_{_{AB}}P_{M}\rho _{AB}}%
\right) ^{q_{_{AB}}}}{q_{_{AB}}}  \label{Pout_asympt}
\end{equation}%
where $p_{_{AB}}=\max \left( \alpha _{_{AB}},\beta _{_{AB}}\right) $ and $%
q_{_{AB}}=\min \left( \alpha _{_{AB}},\beta _{_{AB}}\right) $.
\end{theorem}

\begin{proof}
A detailed proof is provided in Appendix I.
\end{proof}

It should be noted that in the analysis that follows it is assumed that $%
\left( \alpha _{_{AB}}-\beta _{_{AB}}\right) \notin \mathbb{Z}$ holds for
every possible FSO link. Although this condition may seem restrictive, it
can be relaxed in practical applications by inserting an infinitely small
perturbation term $\varepsilon $, so that $\left( \alpha _{_{AB}}-\beta
_{_{AB}}+\varepsilon \right) \notin \mathbb{Z}$, when $\left( \alpha
_{_{AB}}-\beta _{_{AB}}\right) \in \mathbb{Z}$ \cite{J:SchoberTCOM_GG}.

\subsection{Outage Probability of All-Active Relaying}

In this scheme an outage occurs when either the decoding set $\mathbf{D}$ is
empty or the SNR of the multiple-input single-output link between the
decoding relays and the destination falls below the outage threshold. Hence,
the outage probability of this scheme can be expressed as \cite[Eq. (30)]%
{J:UysalRelays}%
\begin{equation}
P_{out}=\sum_{n=1}^{2^{N}}\Pr \left\{ \sum\limits_{m\in S\left( n\right)
}\rho _{_{R_{m}D}}h_{_{R_{m}D}}<\frac{1}{P_{M}}\right\} \Pr \left\{ S\left(
n\right) \right\}  \label{EGC}
\end{equation}%
where $S\left( n\right) $ denotes the $n$th possible decoding set, $2^{N}$
is the total number of decoding sets, and $\Pr \left\{ S\left( n\right)
\right\} $ is the probability of event $\left\{ \mathbf{D}=S\left( n\right)
\right\} $ given by%
\begin{eqnarray}
\Pr \left\{ S\left( n\right) \right\} &=&\prod_{m\in S\left( n\right) }\Pr
\left\{ \gamma _{_{SR_{m}}}>\gamma _{th}\right\} \prod_{m\notin S\left(
n\right) }\Pr \left\{ \gamma _{_{SR_{m}}}<\gamma _{th}\right\}  \notag \\
&=&\prod_{m\in S\left( n\right) }\left( 1-\Pr \left\{ h_{_{SR_{m}}}<\frac{1}{%
\rho _{_{SR_{m}}}P_{M}}\right\} \right) \prod_{m\notin S\left( n\right) }\Pr
\left\{ h_{_{SR_{m}}}<\frac{1}{\rho _{_{SR_{m}}}P_{M}}\right\} .
\label{Decod}
\end{eqnarray}

\subsubsection{Exact Analysis}

In order to evaluate (\ref{EGC}), the cdf of the sum of weighted
non-identical Gamma-Gamma variates, $h_{_{S\left( n\right)
}}=\sum\limits_{m\in S\left( n\right) }\rho _{_{_{R_{m}D}}}h_{_{R_{m}D}}$,
needs to be derived first. However, to the best of the authors' knowledge,
there are no closed-form analytical expressions for the exact distribution
of the sum of non-identical Gamma-Gamma variates. Therefore, the numerical
method of \cite[Eq. (9.186)]{B:Alouini}, which is based on the moment
generating function (MGF) approach, is applied and thus the cdf of $%
h_{_{S\left( n\right) }}$, denoted as $F_{h_{_{S\left( n\right) }}}\left(
\cdot \right) $, is evaluated via%
\begin{eqnarray}
F_{h_{_{S\left( n\right) }}}\left( x\right) &=&\frac{2^{-K}\exp \left( \frac{%
A}{2}\right) }{x}\sum_{k=1}^{K}\left(
\begin{array}{c}
K \\
k%
\end{array}%
\right) \left( \frac{1}{2}\mathrm{Re}\left\{ \frac{\prod\limits_{m\in
S\left( n\right) }\left( \mathcal{M}_{_{R_{m}D}}\left( -\frac{A}{2x}\right)
\right) }{\frac{A}{2x}}\right\} \right.  \notag \\
&&\left. +\sum_{l=1}^{L+k}\left( -1\right) ^{l}\mathrm{Re}\left\{ \frac{%
\prod\limits_{m\in S\left( n\right) }\left( \mathcal{M}_{_{R_{m}D}}\left( -%
\frac{A+j2\pi l}{2x}\right) \right) }{\frac{A+j2\pi l}{2x}}\right\} \right)
\label{Out_sum}
\end{eqnarray}%
where $\mathcal{M}_{_{R_{m}D}}\left( \cdot \right) $ is the MGF of the
channel gain of the $R_{m}D$ FSO link given by \cite[Eq. (4)]{J:bithas},
while the parameters $A$, $K$, $L$ are calculated based on the numerical
error term obtained by \cite[Eq. (9.187)]{B:Alouini}.

\begin{theorem}
The outage probability of the all-active relaying protocol in Gamma-Gamma
fading is given by%
\begin{eqnarray}
P_{out} &=&\sum_{n=1}^{2^{N}}\prod_{m\in S\left( n\right) }\left( 1-\frac{1}{%
\Gamma \left( \alpha _{_{SR_{m}}}\right) \Gamma \left( \beta
_{_{SR_{m}}}\right) }G_{1,3}^{2,1}\left[ \frac{\frac{\alpha
_{_{SR_{m}}}\beta _{_{SR_{m}}}}{\bar{h}_{_{SR_{m}}}}}{\rho _{_{SR_{m}}}P_{M}}%
\left\vert
\begin{array}{c}
1 \\
\alpha _{_{SR_{m}}},\beta _{_{SR_{m}}},0 \\
\end{array}%
\right. \right] \right)  \notag \\
&&\times \prod_{m\notin S\left( n\right) }\frac{1}{\Gamma \left( \alpha
_{_{SR_{m}}}\right) \Gamma \left( \beta _{_{SR_{m}}}\right) }G_{1,3}^{2,1}%
\left[ \frac{\frac{\alpha _{_{SR_{m}}}\beta _{_{SR_{m}}}}{\bar{h}_{_{SR_{m}}}%
}}{\rho _{_{SR_{m}}}P_{M}}\left\vert
\begin{array}{c}
1 \\
\alpha _{_{SR_{m}}},\beta _{_{SR_{m}}},0 \\
\end{array}%
\right. \right] F_{h_{_{S\left( n\right) }}}\left( \frac{1}{P_{M}}\right) .
\label{EGC_4}
\end{eqnarray}
\end{theorem}

\begin{proof}
The proof follows straightforwardly by combining (\ref{EGC}) with (\ref%
{Pout_SISO_Meijer}) and (\ref{Out_sum}).
\end{proof}

\subsubsection{Asymptotic Analysis}

In order to gain more physical insights into the performance of the relaying
protocol under consideration, we further consider the high power margin
regime, i.e., when $P_{M}\rightarrow \infty $. In order to perform this
analysis, an asymptotic expression for $F_{h_{_{S\left( n\right) }}}\left(
\cdot \right) $ needs to be derived first.

\begin{lemma}
For high values of power margin, the cdf for the weighted sum of
non-identical Gamma-Gamma variates that corresponds to decoding set $S\left(
n\right) $, $h_{_{S\left( n\right) }}$, can be approximated as%
\begin{equation}
F_{h_{_{S\left( n\right) }}}\left( x\right) \approx \frac{\prod\limits_{m\in
S\left( n\right) }\left( \frac{\alpha _{_{R_{m}D}}\beta _{_{R_{m}D}}}{\bar{h}%
_{_{R_{m}D}}\rho _{_{R_{m}D}}}\right) ^{q_{_{R_{m}D}}}\frac{\Gamma \left(
q_{_{R_{m}D}}\right) \Gamma \left( p_{_{R_{m}D}}-q_{_{R_{m}D}}\right) }{%
\Gamma \left( \alpha _{_{R_{m}D}}\right) \Gamma \left( \beta
_{_{R_{m}D}}\right) }}{\left( \sum\limits_{m\in S\left( n\right)
}q_{_{R_{m}D}}\right) \Gamma \left( \sum\limits_{m\in S\left( n\right)
}q_{_{R_{m}D}}\right) }x^{\left( \sum\limits_{m\in S\left( n\right)
}q_{_{R_{m}D}}\right) }.  \label{cdf_sum}
\end{equation}
\end{lemma}

\begin{proof}
A detailed proof is provided in Appendix II.
\end{proof}

The asymptotic expression for the outage probability of the all-active
relaying scheme is given by the following theorem.

\begin{theorem}
For high values of power margin, the outage probability of the all-active
relaying scheme in Gamma-Gamma fading can be approximated by
\begin{eqnarray}
P_{out} &\approx &\sum\limits_{n=1}^{2^{N}}\prod\limits_{m\notin S\left(
n\right) }\left( \frac{\pi \Gamma \left( p_{_{SR_{m}}}-q_{_{SR_{m}}}\right)
}{\Gamma \left( \alpha _{_{SR_{m}}}\right) \Gamma \left( \beta
_{_{SR_{m}}}\right) }\frac{\left( \frac{\alpha _{_{SR_{m}}}\beta _{_{SR_{m}}}%
}{\bar{h}_{_{SR_{m}}}\rho _{_{SR_{m}}}}\right) ^{q_{_{SR_{m}}}}}{%
q_{_{SR_{m}}}}\right)  \notag \\
&&\times \frac{\prod\limits_{m\in S\left( n\right) }\left( \frac{\alpha
_{_{R_{m}D}}\beta _{_{R_{m}D}}}{\bar{h}_{_{R_{m}D}}\rho _{_{R_{m}D}}}\right)
^{q_{_{R_{m}D}}}\frac{\Gamma \left( q_{_{R_{m}D}}\right) \Gamma \left(
p_{_{R_{m}D}}-q_{_{R_{m}D}}\right) }{\Gamma \left( \alpha
_{_{R_{m}D}}\right) \Gamma \left( \beta _{_{R_{m}D}}\right) }}{\left(
\sum\limits_{m\in S\left( n\right) }q_{_{R_{m}D}}\right) \Gamma \left(
\sum\limits_{m\in S\left( n\right) }q_{_{R_{m}D}}\right) }\left( \frac{1}{%
P_{M}}\right) ^{\left( \sum\limits_{m\notin S\left( n\right)
}q_{_{SR_{m}}}+\sum\limits_{m\in S\left( n\right) }q_{_{R_{m}D}}\right) }.
\label{Pout_all_active_asymp}
\end{eqnarray}
\end{theorem}

\begin{proof}
We first observe that in the high power margin regime, i.e., $%
P_{M}\rightarrow \infty $, (\ref{Decod}) can be approximated by
\begin{equation}
\Pr \left\{ S\left( n\right) \right\} \approx \prod_{m\notin S\left(
n\right) }\Pr \left\{ h_{_{SR_{m}}}<\frac{1}{\rho _{_{SR_{m}}}P_{M}}\right\}
.  \label{Decod3}
\end{equation}%
Hence, by combining (\ref{EGC}) with (\ref{Pout_asympt}), (\ref{cdf_sum}),
and (\ref{Decod3}), the asymptotic expression in (\ref{Pout_all_active_asymp}%
) is obtained. This concludes the proof.
\end{proof}

An important result derived from the asymptotic expression of the previous
theorem is the diversity gain of the transmission protocol under
consideration which is summarized in the ensuing corollary.

\begin{corollary}
For the all-active relaying protocol, the diversity gain of a relay-assisted
FSO system with $N$ relays is given as%
\begin{equation}
G_{d}=\min_{n\in \left\{ 1,...,2^{N}\right\} }\left( \sum\limits_{m\notin
S\left( n\right) }q_{_{SR_{m}}}+\sum\limits_{m\in S\left( n\right)
}q_{_{R_{m}D}}\right) .  \label{div_all_act}
\end{equation}
\end{corollary}

\begin{proof}
We define the diversity order as
\begin{equation}
G_{d}=-\lim_{P_{M}\rightarrow \infty }\frac{\log P_{out}}{\log P_{M}}.
\end{equation}%
Hence, (\ref{div_all_act}) follows straightforwardly by observing that the
term which corresponds to the power of $P_{M}$ with the minimum of $\left(
\sum\limits_{m\notin S\left( n\right) }q_{_{SR_{m}}}+\sum\limits_{m\in
S\left( n\right) }q_{_{R_{m}D}}\right) $ dominates in the sum of (\ref%
{Pout_all_active_asymp}), when $P_{M}\rightarrow \infty $.
\end{proof}

\subsection{Outage Probability of Select-Max Relaying}

In the select-max protocol a single relay out of the $N$ available relays is
selected according to the selection rule in (\ref{sel_max}). Hence, the
outage probability of the relaying scheme under consideration is given by
\begin{equation}
P_{out}=P_{out}\left\{ R_{1}\cap ...\cap R_{N}\right\}
=\prod\limits_{b=1}^{N}P_{out}\left\{ R_{b}\right\}  \label{sel_max1}
\end{equation}%
where $P_{out}\left\{ R_{b}\right\} $ denotes the probability of outage when
only relay $R_{b}$ is active. Given that $R_{b}$ is active, an outage occurs
when either $R_{b}$ or $D$ have not decoded the information successfully,
i.e.,%
\begin{eqnarray}
P_{out}\left\{ R_{b}\right\} &=&\Pr \left\{ \left( \gamma
_{_{SR_{b}}}<\gamma _{th}\right) \cup \left( \gamma _{_{R_{b}D}}<\gamma
_{th}\right) \right\}  \notag \\
&=&1-\left( 1-\Pr \left\{ \tilde{h}_{_{SR_{b}}}<\frac{\frac{1}{P_{M}}}{\bar{h%
}_{_{SR_{b}}}\rho _{_{SR_{b}}}}\right\} \right) \left( 1-\Pr \left\{ \tilde{h%
}_{_{R_{b}D}}<\frac{\frac{1}{P_{M}}}{\bar{h}_{_{R_{b}D}}\rho _{_{R_{b}D}}}%
\right\} \right) .  \label{sel_max2}
\end{eqnarray}%
Hence, the probability of outage of the select-max relaying scheme is
obtained by combining (\ref{sel_max1}) with (\ref{sel_max2}).

\subsubsection{Exact Analysis}

The following theorem provides an accurate analytical expression for the
performance evaluation of the select-max relaying scheme.

\begin{theorem}
The probability of outage of a relay-assisted FSO system that employs the
select-max relaying protocol in Gamma-Gamma turbulence-induced fading is
given by
\begin{eqnarray}
P_{out} &=&\prod\limits_{b=1}^{N}\left( 1-\left( 1-\frac{1}{\Gamma \left(
\alpha _{_{SR_{b}}}\right) \Gamma \left( \beta _{_{SR_{b}}}\right) }%
G_{1,3}^{2,1}\left[ \frac{\alpha _{_{SR_{b}}}\beta _{_{SR_{b}}}}{\bar{h}%
_{_{SR_{b}}}\rho _{_{SR_{b}}}P_{M}}\left\vert
\begin{array}{c}
1 \\
\alpha _{_{SR_{b}}},\beta _{_{SR_{b}}},0 \\
\end{array}%
\right. \right] \right) \right.  \notag \\
&&\times \left. \left( 1-\frac{1}{\Gamma \left( \alpha _{_{R_{b}D}}\right)
\Gamma \left( \beta _{_{R_{b}D}}\right) }G_{1,3}^{2,1}\left[ \frac{\alpha
_{_{R_{b}D}}\beta _{_{R_{b}D}}}{\bar{h}_{_{R_{b}D}}\rho _{_{R_{b}D}}P_{M}}%
\left\vert
\begin{array}{c}
1 \\
\alpha _{_{R_{b}D}},\beta _{_{R_{b}D}},0 \\
\end{array}%
\right. \right] \right) \right) .  \label{Sel_max_ex}
\end{eqnarray}
\end{theorem}

\begin{proof}
The proof follows straightforwardly by combining (\ref{sel_max1}) with (\ref%
{sel_max2}) and (\ref{Pout_SISO_Meijer}).
\end{proof}

\subsubsection{Asymptotic Analysis}

In order to gain more physical insights into the performance of the relaying
protocol under consideration, we investigate its asymptotic behavior when $%
P_{M}\rightarrow \infty $, in the ensuing theorem and corollary.

\begin{theorem}
For high values of power margin, the outage probability of the select-max
relaying scheme can be approximated as%
\begin{equation}
P_{out}\approx \prod\limits_{b=1}^{N}\left( \frac{\frac{\Gamma \left(
p_{_{SR_{b}}}-q_{_{SR_{b}}}\right) }{q_{_{SR_{b}}}}}{\Gamma \left( \alpha
_{_{SR_{b}}}\right) \Gamma \left( \beta _{_{SR_{b}}}\right) }\left( \frac{%
\frac{\alpha _{_{SR_{b}}}\beta _{_{SR_{b}}}}{\bar{h}_{_{SR_{b}}}}}{\rho
_{_{SR_{b}}}P_{M}}\right) ^{q_{_{SR_{b}}}}+\frac{\frac{\Gamma \left(
p_{_{R_{b}D}}-q_{_{R_{b}D}}\right) }{q_{_{R_{b}D}}}}{\Gamma \left( \alpha
_{_{R_{b}D}}\right) \Gamma \left( \beta _{_{R_{b}D}}\right) }\left( \frac{%
\frac{\alpha _{_{R_{b}D}}\beta _{_{R_{b}D}}}{\bar{h}_{_{R_{b}D}}}}{\rho
_{_{R_{b}D}}P_{M}}\right) ^{q_{_{R_{b}D}}}\right) .  \label{Sel_max_asym}
\end{equation}
\end{theorem}

\begin{proof}
The proof starts by observing that as $P_{M}\rightarrow \infty $ the
probability of outage given relay $R_{b}$ is active, can be approximated by%
\begin{equation}
P_{out}\left\{ R_{b}\right\} \approx \Pr \left\{ \tilde{h}_{_{SR_{b}}}<\frac{%
\frac{1}{P_{M}}}{\bar{h}_{_{SR_{b}}}\rho _{_{SR_{b}}}}\right\} +\Pr \left\{
\tilde{h}_{_{R_{b}D}}<\frac{\frac{1}{P_{M}}}{\bar{h}_{_{R_{b}D}}\rho
_{_{R_{b}D}}}\right\} .  \label{Pout_R_b}
\end{equation}%
Hence, by combining (\ref{sel_max1}) with (\ref{Pout_asympt}) and (\ref%
{Pout_R_b}), the asymptotic expression in (\ref{Sel_max_asym}) is obtained.
This concludes the proof.
\end{proof}

\begin{corollary}
The diversity gain of a relay-assisted FSO system employing the select-max
relaying protocol and $N$ relays is given as%
\begin{equation}
G_{d}=\sum\limits_{b=1}^{N}\min \left( q_{_{SR_{b}}},q_{_{R_{b}D}}\right) .
\label{Div_ord_sel_max}
\end{equation}
\end{corollary}

\begin{proof}
The proof follows straightforwardly from (\ref{Sel_max_asym}). When $%
P_{M}\rightarrow \infty $, the term that corresponds to the power of $P_{M}$
with the minimum of $\left( q_{_{SR_{b}}},q_{_{R_{b}D}}\right) $ dominates
in the sum inside the product. Hence, after taking the product of the
dominating terms, the diversity order is obtained.
\end{proof}

\subsection{Outage Probability of DSSC Relaying}

In the DSSC protocol, the selection of the single relay which takes part in
the communication is based on (\ref{DSSC_1}) and (\ref{DSSC_2}). Hence, an
outage occurs when there is an outage either in the end-to-end link of the
first relay, given that the first relay is selected in the $j$th
transmission slot, or in the end-to-end link of the second relay, given that
the second relay is selected in the $j$th transmission slot, i.e.,
\begin{equation}
P_{out}=\Pr \left\{ \left( R_{b}^{j}=R_{1}\right) \cap \left( \gamma
_{1}^{j}<\gamma _{th}\right) \right\} +\Pr \left\{ \left(
R_{b}^{j}=R_{2}\right) \cap \left( \gamma _{2}^{j}<\gamma _{th}\right)
\right\}  \label{P_DSSC_1}
\end{equation}%
Following the analysis of \cite[Sec. (9.9.1.2)]{B:Alouini}, the above
equation can be rewritten as
\begin{equation}
P_{out}=\left\{
\begin{array}{c}
\frac{\Pr \left\{ \gamma _{1}<T\right\} \Pr \left\{ \gamma _{2}<T\right\} }{%
\Pr \left\{ \gamma _{1}<T\right\} +\Pr \left\{ \gamma _{2}<T\right\} }\left(
\Pr \left\{ \gamma _{1}<\gamma _{th}\right\} +\Pr \left\{ \gamma _{2}<\gamma
_{th}\right\} \right) ,~~~~~~~\gamma _{th}\leq T \\
\multicolumn{1}{l}{\frac{\Pr \left\{ \gamma _{1}<T\right\} \Pr \left\{
\gamma _{2}<T\right\} }{\Pr \left\{ \gamma _{1}<T\right\} +\Pr \left\{
\gamma _{2}<T\right\} }\left( \Pr \left\{ \gamma _{1}<\gamma _{th}\right\}
+\Pr \left\{ \gamma _{2}<\gamma _{th}\right\} -2\right)} \\
\multicolumn{1}{r}{+\frac{\left( \Pr \left\{ \gamma _{1}<\gamma
_{th}\right\} \Pr \left\{ \gamma _{2}<T\right\} +\Pr \left\{ \gamma
_{1}<T\right\} \Pr \left\{ \gamma _{2}<\gamma _{th}\right\} \right) }{\Pr
\left\{ \gamma _{1}<T\right\} +\Pr \left\{ \gamma _{2}<T\right\} },~~\gamma
_{th}>T}%
\end{array}%
\right.  \label{P_DSSC_3}
\end{equation}

\subsubsection{Exact Analysis}

The following theorem provides an accurate analytical expression for the
performance of the DSSC relaying scheme.

\begin{theorem}
The probability of outage of a relay-assisted FSO system that employs the
DSSC relaying protocol is given by%
\begin{equation}
P_{out}=\left\{
\begin{array}{c}
\frac{F_{h_{1}}\left( \frac{1}{\tilde{T}}\right) F_{h_{2}}\left( \frac{1}{%
\tilde{T}}\right) }{F_{h_{1}}\left( \frac{1}{\tilde{T}}\right)
+F_{h_{2}}\left( \frac{1}{\tilde{T}}\right) }\left( F_{h_{1}}\left( \frac{1}{%
P_{M}}\right) +F_{h_{2}}\left( \frac{1}{P_{M}}\right) \right) ,~~~~~~~\bar{T}%
\leq P_{M} \\
\multicolumn{1}{l}{\frac{F_{h_{1}}\left( \frac{1}{\tilde{T}}\right)
F_{h_{2}}\left( \frac{1}{\tilde{T}}\right) }{F_{h_{1}}\left( \frac{1}{\tilde{%
T}}\right) +F_{h_{2}}\left( \frac{1}{\tilde{T}}\right) }\left(
F_{h_{1}}\left( \frac{1}{P_{M}}\right) +F_{h_{2}}\left( \frac{1}{P_{M}}%
\right) -2\right)} \\
\multicolumn{1}{r}{+\frac{\left( F_{h_{1}}\left( \frac{1}{P_{M}}\right)
F_{h_{2}}\left( \frac{1}{\tilde{T}}\right) +F_{h_{1}}\left( \frac{1}{\tilde{T%
}}\right) F_{h_{2}}\left( \frac{1}{P_{M}}\right) \right) }{F_{h_{1}}\left(
\frac{1}{\tilde{T}}\right) +F_{h_{2}}\left( \frac{1}{\tilde{T}}\right) },~~%
\bar{T}>P_{M}}%
\end{array}%
\right.  \label{DSSC_exact}
\end{equation}%
where
\begin{eqnarray}
F_{h_{i}}\left( x\right) &=&1-\left( 1-\frac{1}{\Gamma \left( \alpha
_{_{SR_{i}}}\right) \Gamma \left( \beta _{_{SR_{i}}}\right) }G_{1,3}^{2,1}%
\left[ \frac{\alpha _{_{SR_{i}}}\beta _{_{SR_{i}}}}{\bar{h}_{_{SR_{i}}}\rho
_{_{SR_{i}}}x}\left\vert
\begin{array}{c}
1 \\
\alpha _{_{SR_{i}}},\beta _{_{SR_{i}}},0 \\
\end{array}%
\right. \right] \right)  \notag \\
&&\times \left( 1-\frac{1}{\Gamma \left( \alpha _{_{R_{i}D}}\right) \Gamma
\left( \beta _{_{R_{i}D}}\right) }G_{1,3}^{2,1}\left[ \frac{\alpha
_{_{R_{i}D}}\beta _{_{R_{i}D}}}{\bar{h}_{_{R_{i}D}}\rho _{_{R_{i}D}}x}%
\left\vert
\begin{array}{c}
1 \\
\alpha _{_{R_{i}D}},\beta _{_{R_{i}D}},0 \\
\end{array}%
\right. \right] \right)  \label{cdf_min}
\end{eqnarray}%
and $\bar{T}=\frac{\eta T_{b}P_{t}}{\sqrt{N_{0}T}}$.
\end{theorem}

\begin{proof}
We first note that the cdf of the $\min $ equivalent SNR defined in (\ref%
{min}) can be expressed according to \cite[pp. 141]{B:Pap} as%
\begin{equation}
F_{\gamma _{i}}\left( x\right) =1-\left( 1-F_{\gamma _{_{SR_{i}}}}\left(
x\right) \right) \left( 1-F_{\gamma _{_{R_{i}D}}}\left( x\right) \right)
\end{equation}%
which is equivalent, after a variable transformation, to%
\begin{equation}
F_{\gamma _{i}}\left( x\right) =1-\left( 1-P_{out,SR_{i}}\left( \frac{\sqrt{%
N_{0}x}}{\rho _{_{SR_{i}}}\eta T_{b}P_{t}}\right) \right) \left(
1-P_{out,R_{i}D}\left( \frac{\sqrt{N_{0}x}}{\rho _{_{R_{i}D}}\eta T_{b}P_{t}}%
\right) \right)
\end{equation}%
Using (\ref{Pout_SISO}) the cdf of the $\min $ equivalent channel gain is
derived as (\ref{cdf_min}) and, thus, according to (\ref{P_DSSC_3}), (\ref%
{DSSC_exact}) is obtained. This concludes the proof.
\end{proof}

\begin{corollary}
The performance of the DSSC relaying scheme is minimized when $\bar{T}=P_{M}$
and in that case it becomes equal to that of the select-max scheme for two
relays, $R_{1}$ and $R_{2}$.
\end{corollary}

\begin{proof}
Following the analysis in \cite[(Ch. 9.9.1.1)]{B:Alouini}, the performance
of DSSC relaying is minimized when $\bar{T}=P_{M}$ is set. In that case, (%
\ref{DSSC_exact}) yields%
\begin{equation}
P_{out}=\prod\limits_{i=1}^{2}\left( 1-\left( 1-P_{out,SR_{i}}\right) \left(
1-P_{out,R_{i}D}\right) \right)  \label{optimum_DSSC}
\end{equation}%
which is equivalent to (\ref{Sel_max_ex}) when $N=2$. This concludes the
proof.
\end{proof}

\subsubsection{Asymptotic Analysis}

In order to gain more physical insights of the DSSC protocol with optimized
threshold, we investigate the asymptotic behavior of its performance when $%
P_{M}\rightarrow \infty $ in the ensuing corrolary.

\begin{corollary}
The minimum outage probability of the DSSC relaying protocol can be
approximated at the high power margin regime, by (\ref{Sel_max_asym}) with $%
N=2$, and the maximum achieved diversity gain is given by%
\begin{equation}
G_{d}=\sum\limits_{i=1}^{2}\min \left( q_{_{SR_{i}}},q_{_{R_{i}D}}\right) .
\label{Div_order_DSSC_opt}
\end{equation}
\end{corollary}

\begin{proof}
The proof follows straightforwardly by combining (\ref{optimum_DSSC}) with (%
\ref{Sel_max_asym}) and (\ref{Div_ord_sel_max}).
\end{proof}

\section{Optimal Power Allocation\label{OPA}}

In this section, we are interested in optimizing the optical power resources
in both the $S$-$R_{i}$ and $R_{i}$-$D$ links, in order to minimize the
outage probability of the relay-assisted FSO system for a given total
optical power. Hence, in the following, we optimize the parameters $\rho
_{_{SR_{i}}}$ and $\rho _{_{R_{i}D}}$ for each of the relaying protocols
under consideration.

\subsection{Power Allocation in All-Active Protocol}

Since in the all-active scheme the power is divided among all the underlying
links, the minimization of its outage probability is subject to two
constraints; the total power budget of all links is equal to $P_{t}$ and the
optical power emitted from each transmitter is less than $P_{t}$.
Consequently, the optimum power allocation can be found by solving the
following optimization problem%
\begin{equation}
\begin{array}{c}
\min P_{out} \\
\text{subject to }\left\{
\begin{array}{c}
\sum_{m=1}^{N}\left( \rho _{_{SR_{m}}}+\rho _{_{R_{m}D}}\right) =1 \\
0<\rho _{_{SR_{m}}}\leq 1,~~m=1,...N \\
0<\rho _{_{R_{m}D}}\leq 1,~~m=1,...N%
\end{array}%
\right.%
\end{array}
\label{all_act_pow_alloc}
\end{equation}%
where $P_{out}$ is given by (\ref{EGC_4}). It should be noted that that the
above optimization problem is convex problem. This can be explained as
follows. Since the objective function is an outage probability, it is convex
according to \cite{J:Hasna}. Furthermore, since all the constraints are
linear, they form a convex set \cite{B:Nocedal}, which leads to a convex
optimization problem and, thus, a unique optimal solution.

Using the exact outage expression in (\ref{EGC_4}), it is difficult to find
the optimum solution for the problem in (\ref{all_act_pow_alloc}), even with
numerical methods, due to the involvement of the Meijer's G-functions.
Therefore, the asymptotic expression of (\ref{Pout_all_active_asymp}) is
used as objective function instead and hence the optimization problem is
reformulated as%
\begin{equation}
\begin{array}{r}
\min \left( \sum\limits_{n=1}^{2^{N}}\frac{\left( \frac{1}{P_{M}}\right)
^{\left( \sum\limits_{m\notin S\left( n\right)
}q_{_{SR_{m}}}+\sum\limits_{m\in S\left( n\right) }q_{_{R_{m}D}}\right) }}{%
\left( \sum\limits_{m\in S\left( n\right) }q_{_{R_{m}D}}\right) \Gamma
\left( \sum\limits_{m\in S\left( n\right) }q_{_{R_{m}D}}\right) }%
\prod\limits_{m\notin S\left( n\right) }\left( \frac{\pi \Gamma \left(
p_{_{SR_{m}}}-q_{_{SR_{m}}}\right) \left( \frac{\alpha _{_{SR_{m}}}\beta
_{_{SR_{m}}}}{\bar{h}_{_{SR_{m}}}\rho _{_{SR_{m}}}}\right) ^{q_{_{SR_{m}}}}}{%
\Gamma \left( \alpha _{_{SR_{m}}}\right) \Gamma \left( \beta
_{_{SR_{m}}}\right) q_{_{SR_{m}}}}\right) \right. \\
\left. \times \prod\limits_{m\in S\left( n\right) }\left( \frac{\left( \frac{%
\alpha _{_{R_{m}D}}\beta _{_{R_{m}D}}}{\bar{h}_{_{R_{m}D}}\rho _{_{R_{m}D}}}%
\right) ^{q_{_{R_{m}D}}}\Gamma \left( q_{_{R_{m}D}}\right) \Gamma \left(
p_{_{R_{m}D}}-q_{_{R_{m}D}}\right) }{\Gamma \left( \alpha
_{_{R_{m}D}}\right) \Gamma \left( \beta _{_{R_{m}D}}\right) }\right) \right)
\\
\multicolumn{1}{c}{\text{subject to }\left\{
\begin{array}{c}
\sum_{m=1}^{N}\left( \rho _{_{SR_{m}}}+\rho _{_{R_{m}D}}\right) =1 \\
0<\rho _{_{SR_{m}}}\leq 1,~~m=1,...N \\
0<\rho _{_{R_{m}D}}\leq 1,~~m=1,...N%
\end{array}%
\right.}%
\end{array}
\label{all_act_pow_alloc2}
\end{equation}%
which is a geometric program that can be numerically solved using numerical
optimization techniques, such as the interior point method \cite[Sec. 14]%
{B:Nocedal}.

Since the derivation of the exact solution is cumbersome and motivated by
the dependence of the outage probability on the link distance, the following
suboptimal power allocation scheme for all-active relaying is proposed.

\begin{proposition}
For all-active relaying, the fraction of the total optical power which is
allocated to each link is given by
\begin{equation}
\rho _{_{SR_{i}}}=\frac{d_{_{SR_{i}}}}{\sum\limits_{m=1}^{N}\left(
d_{_{SR_{m}}}+d_{_{R_{m}D}}\right) }\text{ and }\rho _{_{R_{i}D}}=\frac{%
d_{_{R_{i}D}}}{\sum\limits_{m=1}^{N}\left(
d_{_{SR_{m}}}+d_{_{R_{m}D}}\right) }  \label{dist_rule_all_act}
\end{equation}%
for the $S$-$R_{i}$ and $R_{i}$-$D$ links, respectively, with $i=1,...,N$.
\end{proposition}

\subsection{Power Allocation in Select-Max Protocol}

Similar to the all-active scheme, the outage probability of the select-max
protocol can also be minimized by optimizing the optical power resources
which are allocated to each of the links. However, in this scheme, both the
objective function and the constraints are different from the problem in (%
\ref{all_act_pow_alloc}).

Since the total optical power is divided only between the $S$-$R_{b}$ and $%
R_{b}$-$D$ links of the active relay, the problem can be formulated as%
\begin{equation}
\begin{array}{c}
\min P_{out}\left\{ R_{b}\right\} \\
\text{subject to }\left\{
\begin{array}{c}
\rho _{_{SR_{b}}}+\rho _{_{R_{b}D}}=1 \\
0<\rho _{_{SR_{b}}}\leq 1 \\
0<\rho _{_{R_{b}D}}\leq 1%
\end{array}%
\right. ,%
\end{array}
\label{Sel_Max_pow_all_1}
\end{equation}%
where $P_{out}\left\{ R_{b}\right\} $ is the the probability of outage when $%
R_{b}$ is active, given by (\ref{sel_max2}). Based on the same reasoning as
in the previous relaying scheme, we conclude that the above optimization
problem is convex, thus, leading to a unique optimal solution.

Due to the involvement of the Meijer's G-functions, it is again difficult to
find the optimum solution if the exact expression in (\ref{sel_max2}) is
used as objective function, even with numerical methods. Therefore, the
asymptotic expression in (\ref{Pout_R_b}) is employed and hence the power
allocation optimization problem is reformulated as%
\begin{equation}
\begin{array}{r}
\min \left( \frac{\Gamma \left( p_{_{SR_{b}}}-q_{_{SR_{b}}}\right) }{\Gamma
\left( \alpha _{_{SR_{b}}}\right) \Gamma \left( \beta _{_{SR_{b}}}\right)
q_{_{SR_{b}}}}\left( \frac{\frac{\alpha _{_{SR_{b}}}\beta _{_{SR_{b}}}}{\bar{%
h}_{_{SR_{b}}}}}{\rho _{_{SR_{b}}}P_{M}}\right) ^{q_{_{SR_{b}}}}+\frac{%
\Gamma \left( p_{_{R_{b}D}}-q_{_{R_{b}D}}\right) }{\Gamma \left( \alpha
_{_{R_{b}D}}\right) \Gamma \left( \beta _{_{R_{b}D}}\right) q_{_{R_{b}D}}}%
\left( \frac{\frac{\alpha _{_{R_{b}D}}\beta _{_{R_{b}D}}}{\bar{h}_{_{R_{b}D}}%
}}{\rho _{_{R_{b}D}}P_{M}}\right) ^{q_{_{R_{b}D}}}\right) \\
\multicolumn{1}{c}{\text{subject to }\left\{
\begin{array}{c}
\rho _{_{SR_{b}}}+\rho _{_{R_{b}D}}=1 \\
0<\rho _{_{SR_{b}}}\leq 1 \\
0<\rho _{_{R_{b}D}}\leq 1%
\end{array}%
\right.}%
\end{array}
\label{Sel_Max_pow_all_2}
\end{equation}

\begin{theorem}
The power allocation parameters that minimize the outage probability of a
relay-assisted FSO system when operating in Gamma-Gamma turbulence induced
fading and employing select-max relaying, are given by
\begin{equation}
\rho _{_{SR_{b}}}=\left( \delta _{_{SR_{b}}}t_{0}\right) ^{\frac{1}{%
q_{_{SR_{b}}+1}}}\text{ and }\rho _{_{R_{b}D}}=\left( \delta
_{_{R_{bD}}}t_{0}\right) ^{\frac{1}{q_{_{R_{b}D}+1}}}  \label{Sel_max_coef}
\end{equation}%
for the $S$-$R_{b}$ and $R_{b}$-$D$ links respectively, where $b=1,...,N$, $%
\delta _{_{SR_{b}}}=\frac{\Gamma \left( p_{_{SR_{b}}}-q_{_{SR_{b}}}\right) }{%
\Gamma \left( \alpha _{_{SR_{b}}}\right) \Gamma \left( \beta
_{_{SR_{b}}}\right) }\left( \frac{\alpha _{_{SR_{b}}}\beta _{_{SR_{b}}}}{%
\bar{h}_{_{SR_{b}}}P_{M}}\right) ^{q_{_{SR_{b}}}}$, $\delta _{_{R_{bD}}}=%
\frac{\Gamma \left( p_{_{R_{b}D}}-q_{_{R_{b}D}}\right) }{\Gamma \left(
\alpha _{_{R_{b}D}}\right) \Gamma \left( \beta _{_{R_{b}D}}\right) }\left(
\frac{\alpha _{_{R_{b}D}}\beta _{_{R_{b}D}}}{\bar{h}_{_{R_{b}D}}P_{M}}%
\right) ^{q_{_{R_{b}D}}}$ and $t_{0}\in \left[ 0,\min \left( \frac{1}{\delta
_{_{SR_{b}}}},\frac{1}{\delta _{_{R_{bD}}}}\right) \right] $ is the unique
real positive root of
\begin{equation}
S\left( t\right) =\delta _{_{SR_{b}}}^{\frac{1}{q_{_{SR_{b}}+1}}}t^{\frac{1}{%
q_{_{SR_{b}}+1}}}+\delta _{_{R_{bD}}}^{\frac{1}{q_{_{R_{b}D}+1}}}t^{\frac{1}{%
q_{_{R_{b}D}+1}}}-1.  \label{polyonym}
\end{equation}
\end{theorem}

\begin{proof}
We first define the Langrangian associated with the optimization problem of (%
\ref{Sel_Max_pow_all_2}) as%
\begin{equation}
\mathcal{J}=\delta _{_{SR_{b}}}\left( \frac{1}{\rho _{_{SR_{b}}}}\right)
^{q_{_{SR_{b}}}}+\delta _{_{R_{bD}}}\left( \frac{1}{\rho _{_{R_{b}D}}}%
\right) ^{q_{_{R_{b}D}}}-\lambda \left( \rho _{_{SR_{b}}}+\rho
_{_{R_{b}D}}\right) .  \label{Langrange}
\end{equation}%
Applying the method of Langrange multipliers, setting $\frac{\partial
\mathcal{J}}{\partial \rho _{_{SR_{b}}}}=0$ and $\frac{\partial \mathcal{J}}{%
\partial \rho _{_{R_{b}D}}}=0$, and using the equality constraint in (\ref%
{Sel_Max_pow_all_2}), it is straightforward to show that the optimum power
allocation coefficients are given by (\ref{Sel_max_coef}), where $t_{0}=-%
\frac{1}{\lambda }$ is the root of $S\left( t\right) $, which has to take
values in the interval $\left[ 0,\min \left( \frac{1}{\delta _{_{SR_{b}}}},%
\frac{1}{\delta _{_{R_{bD}}}}\right) \right] $ in order to satisfy the
inequality constraints in (\ref{Sel_Max_pow_all_2}). It should be noted that
$S\left( t\right) $ has a unique positive root which lies in this interval;
this can be proved by applying the intemediate value theorem of continuous
functions that shows that $S\left( t\right) $ has at least one real positive
root in the interval $\left[ 0,\min \left( \frac{1}{\delta _{_{SR_{b}}}},%
\frac{1}{\delta _{_{R_{bD}}}}\right) \right] $ and the Descartes' rule of
signs \cite{J:And} which shows that this root is unique.
\end{proof}

In order to avoid finding the root of the polynomial $S\left( t\right) $ and
motivated by the dependence of the outage probability on the link distance,
the following suboptimal power allocation scheme for select-max relaying is
proposed.

\begin{proposition}
The fraction of the total optical power which is allocated to each link when
the select-max relaying scheme is employed is chosen as%
\begin{equation}
\rho _{_{SR_{b}}}=\frac{d_{_{SR_{b}}}}{d_{_{SR_{b}}}+d_{_{R_{b}D}}}\text{
and }\rho _{_{R_{b}D}}=\frac{d_{_{R_{b}D}}}{d_{_{SR_{b}}}+d_{_{R_{b}D}}}
\label{Sel_Max_coef_2}
\end{equation}%
for the $S$-$R_{b}$ and $R_{b}$-$D$ links, respectively, with $b=1,...,N$.
\end{proposition}

\subsection{Power Allocation in DSSC Protocol}

Since a single relay is activated in each transmission slot by the DSSC
protocol, the optimum power allocation scheme is obtained by minimizing the
outage probability of the active end-to-end path. Hence, the optimization
problem that has to be solved in this case is equivalent to the problem in (%
\ref{Sel_Max_pow_all_2}) and, therefore, the optimum and the suboptimum
power allocation schemes of the select-max protocol can also be employed for
DSSC relaying.

\section{Numerical Results and Discussion\label{NR}}

In this section, we illustrate numerical results for the outage performance
of the considered relaying protocols, using the derived analytical
expressions. In the following, we consider a relay-assisted FSO system with $%
\lambda =1550$ nm and transmit and receive aperture diameters of $%
D_{R}=D_{T}=20$ cm. Furthermore, we assume clear weather conditions with
visibility of $10$ km, which correspond to a weather-dependent attenuation
coefficient of $v=0.1$ $\frac{1}{\text{km}}$ and an index of refraction
structure parameter of $C_{n}^{2}=2\times 10^{-14}$ m$^{-\frac{2}{3}}$.

Fig. \ref{Fig:1} depicts the outage performance of the presented relaying
protocols for various numbers of relays, when the link distance is identical
for all $S$-$R_{i}$ and $R_{i}$-$D$ links and the optical power is equally
divided between the active relays. Specifically, analytical results for the
outage probability of a relay-assisted FSO system with a link distance of $2$
km are plotted, as functions of the power margin for $N=2,3,4$ relays, using
the exact and the asymptotic outage expressions for each of the considered
relaying protocols. We assumed $\rho _{_{SR_{i}}}=\rho _{_{R_{i}D}}=\frac{1}{%
2N}$ for the all-active and $\rho _{_{SR_{i}}}=\rho _{_{R_{i}D}}=\frac{1}{2}$
for the select-max and DSSC protocols, respectively. As benchmarks, Monte
Carlo (MC) simulation results and the performance of an FSO system with $N=1$%
, which is independent of the employed relaying protocol, are also
illustrated in Fig. \ref{Fig:1}. As can be observed, there is an excellent
match between simulation and analytical results for every value of $N$,
verifying the presented theoretical analysis. Moreover, it is obvious that
the select-max relaying scheme has a better performance compared to the
all-active scheme in every case examined (performance gains of 2, 4, and 5
dB are observed for $N=2$, $3,$ and $4,$ respectively). This result is
intuitively pleasing, since the select-max protocol selects in each
transmission slot the best end-to-end path out of the $N$ available paths
and allocates the total available optical power only to this path.
Furthermore, when increasing the number of relays in the select-max and
all-active protocols, it is observed that the outage performance is
significantly improved with respect to the single relay FSO system. In
contrast, although the DSSC scheme with the optimum threshold offers
significant performance improvement for $N=2$ (its performance is identical
with the select-max performance of $N=2$), it remains unaffected by the
increase of the number of relays.

Fig. \ref{Fig:Comparison1} depicts the outage performance of a
relay-assisted FSO system employing the presented protocols and assuming
different distances for each of the $S$-$R_{i}$ and $R_{i}$-$D$ FSO links.
Specifically, two different system configurations are investigated. In the
first system configuration, $N=2$ and the link distances are given by
vectors $d_{_{SR}}=\left\{ 2,1.5\right\} $ and $d_{_{RD}}=\left\{
1,2.5\right\} $, with the elements of the vectors respresenting the
distances (in km) of the $S$-$R_{i}$ and $R_{i}$-$D$ links respectively,
while in the second configuration $N=3$, and the link distances are given by
$d_{_{SR}}=\left\{ 2,1.5,1\right\} $ and $d_{_{RD}}=\left\{ 1,2.5,3\right\} $%
. Fig. \ref{Fig:Comparison1} reveals that the select-max relaying scheme
offers significant performance gains compared to the all-active scheme, also
for non-equal link distances. In particular, in the first configuration a
gain of 2.5 dB compared to the all-active scheme is offered, while in the
second configuration the offered gain is 3 dB. Furthermore, it can be easily
observed that although in the second configuration the number of relays has
been increased and the performance of both all-active and select-max
relaying has been improved, DSSC with optimized threshold remains unaffected
by this increase and its performance remains identical with the performance
for the first configuration. This was expected, since DSSC uses only two
end-to-end paths (those with the minimum end-to-end distance) and,
therefore, the addition of extra paths with larger end-to-end distance will
not improve the performance of this protocol. Finally, we note that
simulation and analytical results are again in excellent agreement.

Fig. \ref{Fig:pow_alloc} illustrates the effect of power allocation in
relay-assisted FSO systems employing the relaying protocols under
consideration. Specifically, the performaces of the optimum and the proposed
sub-optimum power allocation schemes, obtained by solving (\ref%
{all_act_pow_alloc2}), (\ref{Sel_Max_pow_all_2}) and from the empirical
rules of (\ref{dist_rule_all_act}), (\ref{Sel_Max_coef_2}), respectively,
are presented along with the equal power allocation, when the second system
configuration of Fig. \ref{Fig:Comparison1} is considered. It is obvious
from Fig. \ref{Fig:pow_alloc} that optimized power allocation offers
significant performance gains compared to equal power allocation,
irrespective of the employed relaying protocol. This was expected, since
both the path loss and turbulence strength are distance-dependent in FSO
links, and, hence, power allocation schemes that take into consideration the
distances of the underlying links, can significantly improve system
performance. Furthermore, it is observed that even the simple sub-optimum
power allocations schemes lead to substancial performance improvements
compared to equal power allocation. Taking into consideration that the
parameters for these schemes can be easily obtained, based only on the link
distances (for most practical FSO applications, the link distances are fixed
and, thus, are known a-priori at both the transmitter and relays), the
proposed sub-optimum power allocation can be considered as a less complex
alternative to optimum power allocation.

\section{Conclusions\label{Con}}

We investigated several transmission protocols for relay-assisted FSO
systems without direct link between the source and the destination for the
Gamma-Gamma channel model. Alternative protocols to the all-active relaying
scheme were proposed, which activate only a single relay in each
transmission slot. Thus, considerable benefits in terms of implementation
complexity are resulted, since the need for synchronizing the relays'
transmissions in order for the FSO signals to arrive simultaneously at the
destination is avoided. In particular, two different types of relay
selection protocols were proposed: select-max and DSSC. Select-max relaying
offers significant performance gains compared to the all-active scheme at
the expense of requiring the CSI by all the available links. In contrast,
DSSC relaying requires less CSI than select-max (only from the links used in
the previous transmission slot), however it exploits only the two relays
with the minimum end-to-end distance. Furthermore, based on the derived
outage probability expressions, the problem of allocating the power
resources to the FSO links was addressed, and optimum and sub-optimum
solutions that minimize the system's outage probability were derived for
each considered relaying protocol. Numerical results were provided, which
clearly demonstrated the improvements in the power efficiency offered by the
proposed power allocation schemes.

\section*{Appendix I}

Using the infinite series representation of the Gamma-Gamma pdf \cite[Eqs.
(7), (8)]{J:SchoberTCOM_GG} and since $P_{out,AB}=\int_{0}^{\frac{1}{\bar{h}%
_{_{AB}}\rho _{AB}P_{M}}}f_{\tilde{h}_{_{AB}}}\left( x\right) dx$, the
outage probability for the FSO link between terminals $A$ and $B$ can be
obtained after some basic algebraic manipulations, as%
\begin{equation}
P_{out,AB}=\frac{\frac{\pi }{\sin \left( \pi \left( \alpha _{_{AB}}-\beta
_{_{AB}}\right) \right) }}{\Gamma \left( \alpha _{_{AB}}\right) \Gamma
\left( \beta _{_{AB}}\right) }\sum_{l=0}^{\infty }\left( \frac{\frac{1}{%
l!\left( \beta _{_{AB}}+l\right) }\left( \frac{\alpha _{_{AB}}\beta _{_{AB}}%
}{\bar{h}_{_{AB}}\rho _{AB}P_{M}}\right) ^{\beta _{_{AB}}+l}}{\Gamma \left(
l-\alpha _{_{AB}}+\beta _{_{AB}}+1\right) }-\frac{\frac{1}{l!\left( \alpha
_{_{AB}}+l\right) }\left( \frac{\alpha _{_{AB}}\beta _{_{AB}}}{\bar{h}%
_{AB}\rho _{AB}P_{M}}\right) ^{\alpha _{_{AB}}+l}}{\Gamma \left( l+\alpha
_{_{AB}}-\beta _{_{AB}}+1\right) }\right) .  \label{Pout_SISO_inf}
\end{equation}%
For high values of power margin, i.e., $P_{M}\rightarrow \infty $, the term
for $l=0$ is dominant and hence (\ref{Pout_SISO_inf}) can be approximated by%
\begin{equation}
P_{out,AB}=\frac{\frac{\pi }{\sin \left( \pi \left( \alpha _{_{AB}}-\beta
_{_{AB}}\right) \right) }}{\Gamma \left( \alpha _{_{AB}}\right) \Gamma
\left( \beta _{_{AB}}\right) }\left( \frac{\left( \frac{\alpha _{_{AB}}\beta
_{_{AB}}}{\bar{h}_{_{AB}}\rho _{AB}P_{M}}\right) ^{\beta _{_{AB}}}}{\Gamma
\left( 1-\alpha _{_{AB}}+\beta _{_{AB}}\right) \beta _{_{AB}}}-\frac{\left(
\frac{\alpha _{_{AB}}\beta _{_{AB}}}{\bar{h}_{_{AB}}\rho _{AB}P_{M}}\right)
^{\alpha _{_{AB}}}}{\Gamma \left( 1+\alpha _{_{AB}}-\beta _{_{AB}}\right)
\alpha _{_{AB}}}\right) ,
\end{equation}%
which can be reduced to (\ref{Pout_asympt}), after using the Euler's
reflection formula \cite[Eq. (8.334.3)]{B:Gra_Ryz_Book} and\ introducing $%
p_{_{AB}}$ and $q_{_{AB}}$. This concludes the proof.

\section*{Appendix II}

Based on the infinite series representation of the Gamma-Gamma pdf \cite[%
Eqs. (7), (8)]{J:SchoberTCOM_GG}, the pdf of $\xi _{m}=\rho
_{_{_{R_{m}D}}}h_{_{R_{m}D}}$ can be written as%
\begin{equation}
f_{_{\xi _{m}}}\left( x\right) =\frac{\left( \frac{\alpha _{_{R_{m}D}}\beta
_{_{R_{m}D}}}{\bar{h}_{_{R_{m}D}}\rho _{_{_{R_{m}D}}}}\right)
^{q_{_{R_{m}D}}}\Gamma \left( p_{_{R_{m}D}}-q_{_{R_{m}D}}\right) }{\Gamma
(\alpha _{_{R_{m}D}})\Gamma (\beta _{_{R_{m}D}})}x^{q_{_{R_{m}D}}-1}+O\left(
x^{q_{_{R_{m}D}}}\right) ,
\end{equation}%
where $O\left( \cdot \right) $ represents the least significant terms of an
infinite series as $x\rightarrow \infty $. By taking the Laplace transform
of the above equation, the MGF expression of $\xi _{m}$ can be obtained as%
\begin{equation}
\mathcal{M}_{_{\xi _{m}}}\left( s\right) =\frac{\left( \frac{\alpha
_{_{R_{m}D}}\beta _{_{R_{m}D}}}{\bar{h}_{_{R_{m}D}}\rho _{_{_{R_{m}D}}}}%
\right) ^{q_{_{R_{m}D}}}\Gamma \left( p_{_{R_{m}D}}-q_{_{R_{m}D}}\right) }{%
\Gamma (\alpha _{_{R_{m}D}})\Gamma (\beta _{_{R_{m}D}})}s^{-q_{_{R_{m}D}}}+O%
\left( s^{-q_{_{R_{m}D}}-1}\right) .
\end{equation}%
Hence, the MGF for $h_{_{S\left( n\right) }}=\sum\limits_{m\in S\left(
n\right) }\rho _{_{_{R_{m}D}}}h_{_{R_{m}D}}$ can be written as
\begin{equation}
\mathcal{M}_{_{h_{_{S\left( n\right) }}}}\left( s\right)
=s^{-\sum\limits_{m\in S\left( n\right) }q_{_{R_{m}D}}}\prod\limits_{m\in
S\left( n\right) }\frac{\left( \frac{\alpha _{_{R_{m}D}}\beta _{_{R_{m}D}}}{%
\bar{h}_{_{R_{m}D}}\rho _{_{_{R_{m}D}}}}\right) ^{q_{_{R_{m}D}}}\Gamma
\left( p_{_{R_{m}D}}-q_{_{R_{m}D}}\right) }{\Gamma (\alpha
_{_{R_{m}D}})\Gamma (\beta _{_{R_{m}D}})}+O\left( s^{-\sum\limits_{m\in
S\left( n\right) }q_{_{R_{m}D}}-1}\right) .  \label{mgf_sum}
\end{equation}%
and, by taking the inverse Laplace transform of (\ref{mgf_sum}), an
expression for the pdf of $h_{_{S\left( n\right) }}$ is obtained as
\begin{equation}
f_{_{h_{_{S\left( n\right) }}}}\left( x\right) =\frac{\prod\limits_{m\in
S\left( n\right) }\left( \frac{\alpha _{_{R_{m}D}}\beta _{_{R_{m}D}}}{\bar{h}%
_{_{R_{m}D}}\rho _{_{_{R_{m}D}}}}\right) ^{q_{_{R_{m}D}}}\frac{\Gamma \left(
q_{_{R_{m}D}}\right) \Gamma \left( p_{_{R_{m}D}}-q_{_{R_{m}D}}\right) }{%
\Gamma \left( \alpha _{_{R_{m}D}}\right) \Gamma \left( \beta
_{_{R_{m}D}}\right) }}{\Gamma \left( \sum\limits_{m\in S\left( n\right)
}q_{_{R_{m}D}}\right) }x^{\left( \sum\limits_{m\in S\left( n\right)
}q_{_{R_{m}D}}\right) -1}+O\left( x^{\sum\limits_{m\in S\left( n\right)
}q_{_{R_{m}D}}}\right)  \label{pdf_sum}
\end{equation}%
After some basic algebraic manipulations and keeping only the dominant term,
the asymptotic expression in (\ref{cdf_sum}) is obtained for the cdf of $%
h_{_{S\left( n\right) }}$. This concludes the proof.

\bibliographystyle{IEEEtran}
\bibliography{IEEEabrv,References}
\newpage
\begin{figure}[tbp]
\centering\includegraphics[keepaspectratio,width=7in]{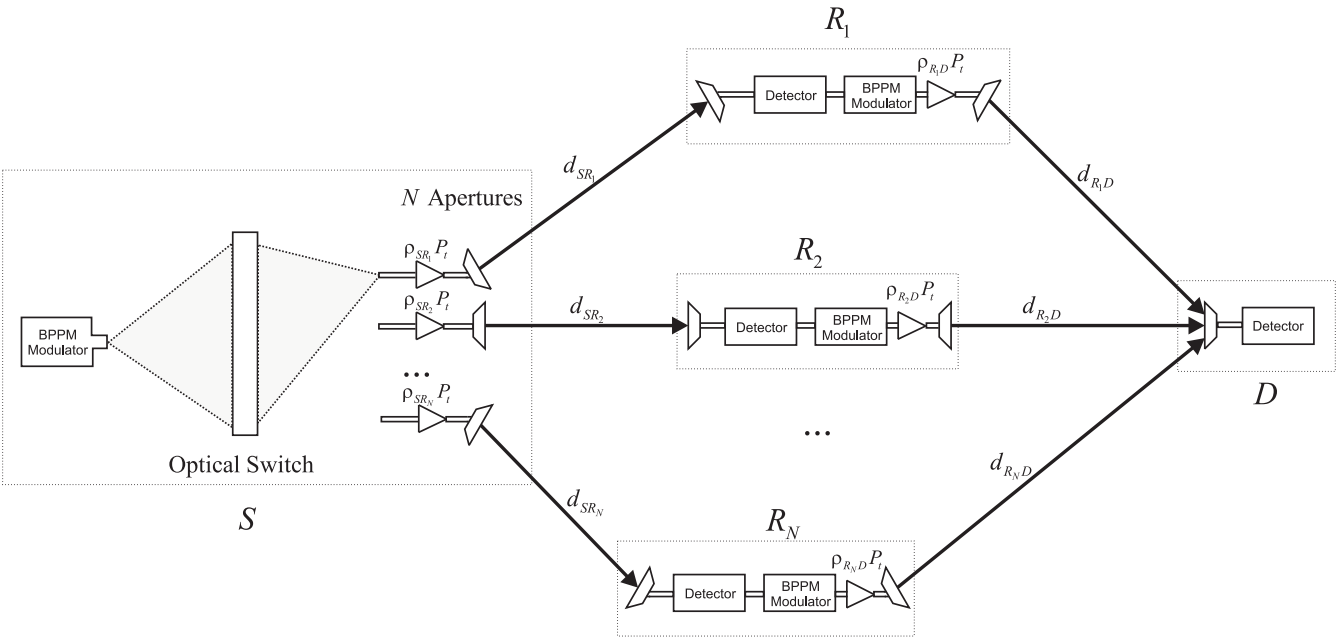}
\caption{The relay-assisted FSO system under consideration.}
\label{Fig:system}
\end{figure}
\newpage
\begin{figure}[tbp]
\centering\includegraphics[keepaspectratio,width=6in]{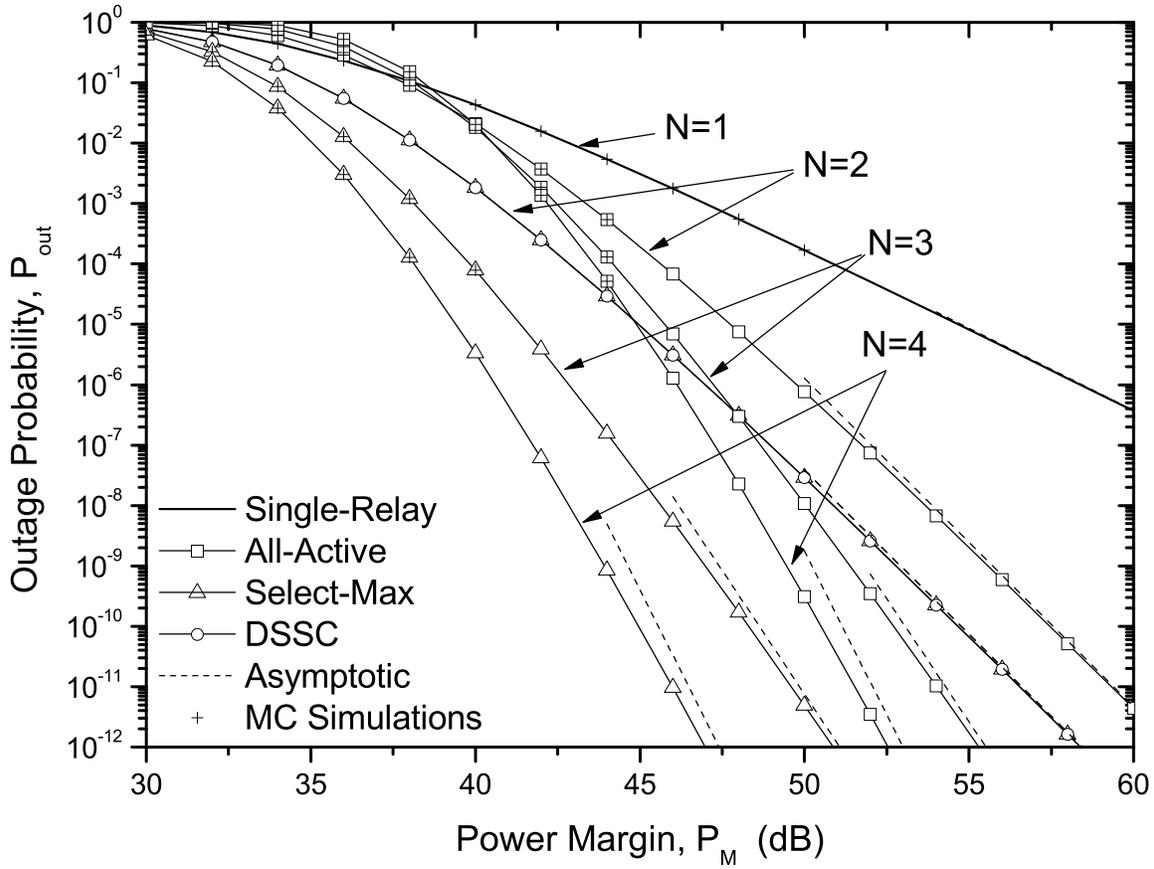}
\caption{Comparison of relaying protocols for a relay-assisted FSO system
with $d_{_{SR_{i}}}=d_{_{R_{i}D}}=2$ km, $i\in \left\{ 1,...,N\right\} $.}
\label{Fig:1}
\end{figure}
\newpage
\begin{figure}[tbp]
\centering\includegraphics[keepaspectratio,width=6in]{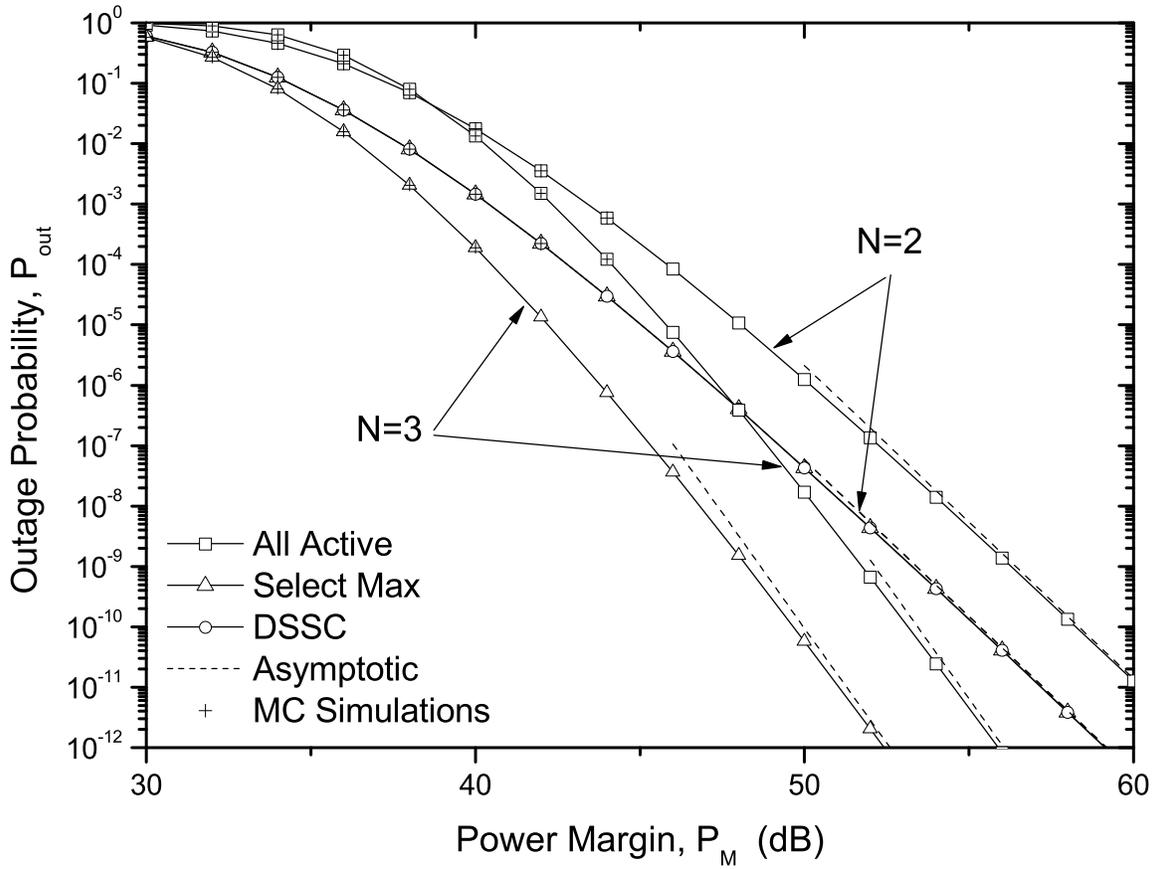}
\caption{Comparison of relaying protocols for different relay-assisted FSO
configurations: $N=2$, $d_{_{SR}}=\left\{ 2,1.5\right\} $, $%
d_{_{RD}}=\left\{ 1,2.5\right\} $ (in km) and $N=3$, $d_{_{SR}}=\left\{
2,1.5,1\right\} $, $d_{_{RD}}=\left\{ 1,2.5,3\right\} $ (in km).}
\label{Fig:Comparison1}
\end{figure}
\newpage
\begin{figure}[tbp]
\centering\includegraphics[keepaspectratio,width=6in]{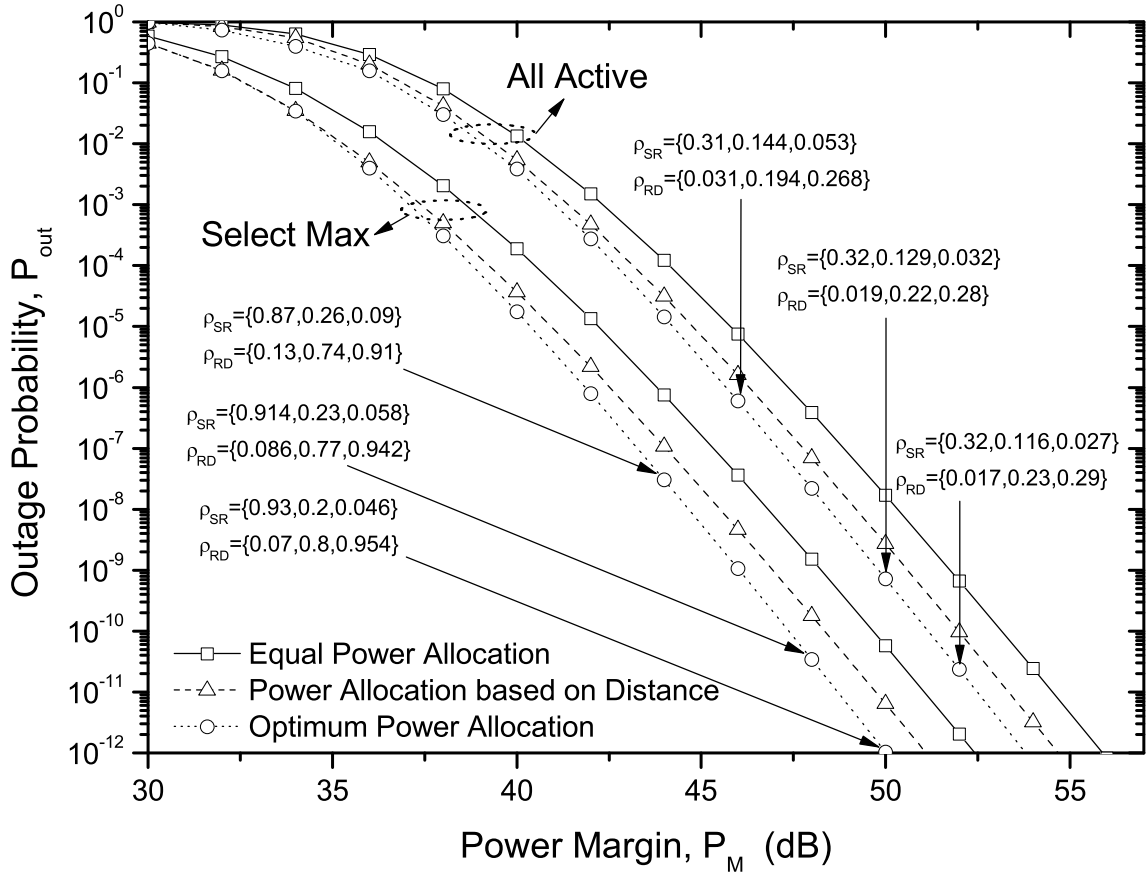}
\caption{Comparison of power allocation schemes for the relaying protocols
under consideration.}
\label{Fig:pow_alloc}
\end{figure}

\end{document}